\documentclass[envcountsame,envcountsect,oribibl,orivec]{llncs} 
\usepackage{ae_logic}

\title{The Complexity of Reasoning for Fragments of Autoepistemic Logic\thanks{Supported in part by DFG grant VO 630/6-2.}}

\author{Nadia Creignou\inst{1} \and Arne Meier\inst{2} \and Michael Thomas\inst{2} \and Heribert Vollmer\inst{2}}

\institute{
  Laboratoire d'Informatique Fondamentale,  CNRS,  Aix-Marseille Universit\'{e}\\
  163, avenue de Luminy, 13288 Marseille, France\\
  \protect\url|creignou@lif.univ-mrs.fr|
\and
  Institut f\"ur Theoretische~Informatik, Gottfried Wilhelm Leibniz Universit\"{a}t\\
  Appelstr.~4, 30167~Hannover, Germany\\
  \protect\url|{meier,thomas,vollmer}@thi.uni-hannover.de|
}

\begin{document}

\maketitle

\begin{abstract}
Autoepistemic logic extends propositional logic by the modal operator $L$. A formula $\varphi$ that is preceded by an $L$ is said to be ``believed''. The logic was introduced by Moore 1985 for modeling an ideally rational agent's behavior and reasoning about his own beliefs. In this paper we analyze all Boolean fragments 
of autoepistemic logic
with respect to the computational complexity of
the three most common decision problems
expansion existence, brave reasoning and cautious reasoning.
As a second contribution we classify the computational complexity of counting the number of stable expansions of a given knowledge base.
To the best of our knowledge this is the first paper analyzing the counting problem for autoepistemic logic.
\end{abstract}

\section{Introduction} \label{sect:intro}

  Non-monotonic logics are among the most important calculi in the area of knowledge representation and reasoning.
  Autoepistemic logic, introduced 1985 by Moore~\cite{moore:85}, is one of the most prominent non-monotonic logic.
  It was originally created to overcome difficulties present in the non-monotonic modal logics proposed by McDermott and Doyle~\cite{mcdermott-doyle:80}, but was also shown to embed several other non-monotonic formalisms such as Reiter's 
  default logic~\cite{reiter:80} or McCarthy's circumscription~\cite{mccarthy:80}.
  
  Autoepistemic logic extends classical logic with a unary modal operator $L$ expressing the beliefs of an ideally rational agent.
  The sentence $L\varphi$ means that the agent can derive $\varphi$ based on its knowledge. 
  To formally capture the set of beliefs of an agent, the notion of \emph{stable expansions} was introduced.
  Stable expansions are defined as the fixed points of an operator deriving the logical consequences of the 
  agent's knowledge and belief. A given knowledge base may admit no or several such stable expansions. 
  Hence, the following questions naturally arise. 
  Given a knowledge base $\Sigma$, does $\Sigma$ admits a stable expansion at all?
  And given a knowledge base $\Sigma$ and a formula $\varphi$, is $\varphi$ contained in at least one (resp.\ all) 
  stable expansion of $\Sigma$.
  
  While all these problems are undecidable for first-order autoepistemic logic, they
 are situated at the second level of the polynomial hierarchy in the propositional case~\cite{gottlob:92}; and thus harder to solve
  than the classical satisfiability or implication problem unless the polynomial hierarchy collapses below its second level.
  This increase in complexity raises the question for the sources of the hardness on the one hand, 
  and for tractable restrictions on the other hand.
  
  In this paper, we study the computational complexity of the three decision problems mentioned above for fragments of autoepistemic logic, given by restricting the propositional part, \emph{i.e.}, by restricting the set of allowed Boolean connectives. 
We bound the complexity of all three reasoning tasks for all finite sets of allowed Boolean functions. 
  This approach has first been taken by Lewis, who showed that the satisfiability problem for (pure) propositional logic is $\NP$-complete if
  the negation of the implication $(x \not\limplies y)$ can be composed from the set of available Boolean connectives, and is polynomial-time solvable in all other cases. Since then, this approach has been applied to a wide range of problems including
  equivalence and implication problems~\cite{rei03,bemethvo08imp}, 
  satisfiability and model checking in modal and temporal logics~\cite{bhss05b,bsssv07,bamuscscscvo09,memuthvo08},
  default logic~\cite{bemethvo08}, 
  circumscription~\cite{thomas09} and abduction \cite{crscth10}.
  
  Our goal is to exhibit fragments of lower complexity which might lead to better algorithms for cases in which the 
  set of Boolean connectives can be restricted. Furthermore we aim to understand the sources of hardness and to provide an understanding which connectives take the role of $(x \not\limplies y)$ 
  in the context of autoepistemic logic.

  Though at first sight, an infinite number of sets $B$ of allowed propositional connectives has to be examined, we prove, making use of results from universal algebra, that essentially only seven cases can occur: (1) $B$ can express all Boolean connectives, (2) $B$ can express all monotone Boolean connectives, (3) $B$ can express all linear connectives, (4) $B$ is equivalent to $\{\lor\}$, (5) $B$ is equivalent to $\{\land\}$, (6) $B$ is equivalent to $\{\neg\}$,  (7) $B$ is empty.
  We first show, extending Gottlob's results,  that the above problems are complete for a 
  class from the second level of the polynomial hierarchy for the cases (1) and (2). 
  In case (4) the complexity drops to completeness for a class from the first level of the hierarchy,
  whereas for (3) the problem becomes solvable in polynomial time while being hard for $\parityL$.
  Finally, for the cases (5) to (7) it even drops down to solvability in logspace.
  
  Besides the decision variant, another natural question is concerned with the number of stable expansions. 
  This refers to the so called \emph{counting problem} for stable expansions.
  Recently, counting problems have gained quite a lot of attention in non-monotonic logics.
  For circumscription, the counting problem (that is, determining the number of minimal models of a propositional formula) has been studied in \cite{dhk05,duhe08}.
  For propositional abduction, a different non-monotonic logic, some complexity results  for the counting problem (that is, computing the number of so called ``solutions'' of a propositional abduction problem) were presented in \cite{hepi07,crscth10}. Algorithms based on bounded treewidth have been proposed in \cite{JPRW08} for the counting problems in abduction and circumscription.
  Here, we consider the complexity of the problem to count the number of stable expansions for a given knowledge base. To the best of our knowledge, this problem is addressed here for the first time. We show that 
  it is $\SharpCoNP$-complete in cases (1) and (2) from the above, 
  drops to $\SharpP$-completeness for the case (4), 
  and is polynomial-time computable in cases (3) and (5) to (7).
  
  The rest of this paper is organized is follows. 
  Sections~\ref{sect:prelim} and~\ref{sect:ae_logic} contain preliminaries and the formal definition of autoepistemic logic.
  In Section~\ref{sect:complexity_results} we classify the complexity of the decision problems mentioned above for all finite 
  sets of allowed Boolean functions. 
  Section~\ref{sect:counting} contains the classification of the problem to
count the number of stable expansions. 
  Finally, Section~\ref{sect:conclusion} concludes with a discussion of the results.

\section{Preliminaries} \label{sect:prelim}

  We use standard notions of complexity theory.
  For decision problem, the arising complexity degrees encompass the classes $\L$,
  $\P$, $\NP$, $\co\NP$, $\SigmaPtwo$ and $\PiPtwo$.
  For more background information, the reader is referred to~\cite{pap94}.
  We furthermore require the class $\parityL$ defined as the class of languages $L$ 
  such that there exists a nondeterministic logspace Turing machine that exhibits an 
  odd number of accepting paths iff $x\in L$ for all $x$~\cite{budaheme92}. 
  It is known that $\L \subseteq \parityL \subseteq \P$.
  Regarding hardness proofs of decision problems, 
  we consider \emph{logspace many-one reductions}, defined as follows:
  a language $A$ is logspace many-one reducible to some language $B$ (written $A \leqlogm B$) if 
  there exists a logspace-computable function $f$ such that $x \in A \iff f(x) \in B$.
  
  In the context of counting problems,
  denote by $\FP$ the set of all functions computable in polynomial time, and for an arbitrary complexity class $\mathcal{C}$,  
  define $\#\mathord{\cdot}\mathcal{C}$ as the class the functions $f$ for which there exists a set $A\in\mathcal{C}$ (the \emph{witness set} for $f$) such that there exists a polynomial $p$ such that for all $(x,y)\in A$, $|y| \leq\bigl|p(x)\bigr|$, and $f(x)=\bigl|\{˛ y \mid (x,y) \in A\}\bigr|$, 
  see \cite{hevo95}.  
  In particular, we make use of the classes $\SharpP=\#\mathord{\cdot}\P$ and $\SharpCoNP$.
  To obtain hardness results for counting problems, we will employ \emph{parsimonious reductions} defined as follows:  
A counting function $f$ parsimoniously reduces to function $h$ if there is a function $g\in\FP$ such that for all $x$, $f(x)=h\bigl(g(x)\bigr)$. Note the analogy to the simple $m$-reductions for decision problems defined above.

  We moreover assume familiarity with propositional logic. 
  As we are going to consider problems parameterized by the set of  Boolean connectives,
  we require some algebraic tools to classify the complexity of the infinitely many arising problems. 
  A \emph{clone} is a set $B$ of Boolean functions that is closed under superposition, 
  \emph{i.e.}, $B$ contains all projections and is closed under arbitrary compositions (see \cite[Chapter~1]{pip97b} or \cite{bcrv03}).
  For a set $B$ of Boolean functions, 
  we denote by $[B]$ the smallest clone containing $B$ and call $B$ a \emph{base} for $[B]$.
  Post classified the lattice of all clones and found a finite base for each clone~\cite{pos41}.  
  A list of all Boolean clones together with a basis for each of them can be found, \emph{e.g.}, in \cite{bcrv03}. 
  In order to introduce the clones relevant to this paper, say that an $n$-ary Boolean functions $f$
  is \emph{monotone} if $a_1 \leq b_1, a_2 \leq b_2, \ldots , a_n \leq b_n$ implies $f(a_1, \ldots , a_n) \leq f(b_1, \ldots , b_n)$, and that $f$ is \emph{linear} if $f \equiv x_1 \xor \cdots \xor x_n \xor c$ for a constant $c \in \{0, 1\}$ and variables $x_1,\ldots, x_n$. 
  The clones relevant to this paper together with their bases are listed in Table~\ref{tab:clones}. 

  \begin{table*}[tb]
    \centering
    \begin{tabular}{c|l|l}
      Name & Definition & Base \\
      \hline
      $\CloneBF$ & All Boolean functions & $\{\land, \neg\}$ \\
      $\CloneM$ & $\{f : f \text{ is monotone}\}$ & $\{\lor, \land, \false, \true\}$ \\
      $\CloneL$ & $\{f : f \text{ is linear}\}$ & $\{ \xor,\true\}$ \\
      $\CloneV$ & $\{f : f \equiv c_0 \lor \bigvee_{i=1}^n c_ix_i \text{ where the $c_i$s are constant}\}$ & $\{ \lor, \false,\true \}$ \\
      $\CloneE$ & $\{f : f \equiv c_0 \land \bigwedge_{i=1}^n c_ix_i \text{ where the $c_i$s are constant}\}$ & $\{ \land, \false, \true \}$ \\
      $\CloneN$ & $\{f : f \text{ depends on at most one variable}\}$ & $\{ \neg,\false,\true\}$ \\
      $\CloneI$ & $\{f : f \text{ is a projection or a constant}\}$ & $\{\id, \false,\true\}$ \\
    \end{tabular}
    \caption{
      \label{tab:clones}
      A list of Boolean clones with definitions and bases.
    }
  \end{table*}

\section{Autoepistemic Logic} \label{sect:ae_logic}

  Autoepistemic logic extends propositional logic by a modal operator $L$ stating that its argument is ``believed''.
  Syntactically, the set of autoepistemic formulae $\allAutoepistemicFormulae$ is defined via 
  $
    \varphi ::= p \mid f(\varphi,\ldots,\varphi) \mid L \varphi,
  $
  where $f$ is a Boolean function and $p$ is a proposition. 
  The consequence relation $\models$ of the underlying propositional logic is extended to $\allAutoepistemicFormulae$ by simply 
  treating $L\varphi$ as an atomic formula.
  An \emph{(autoepistemic) $B$-formula} is an autoepistemic formula using only 
  functions from a finite set $B$ of Boolean functions as connectives.
  The set of all autoepistemic $B$-formulae will be denoted by $\allAutoepistemicFormulae(B)$.
  
  Let $B$ be any finite set of Boolean functions.
  For $\Sigma \subseteq \allAutoepistemicFormulae(B)$, we write $\theorems{\Sigma}$ for the deductive closure of $\Sigma$, 
  \emph{i.e.}, $\theorems{\Sigma} := \{\varphi \mid \Sigma \models \varphi\}$.
  For $\varphi \in \allAutoepistemicFormulae(B)$, let $\SF(\varphi)$ be the set of its subformulae and 
  let $\SF^L(\varphi):=\{L\psi \mid L\psi \in \SF(\varphi)\}$ be the set of its $L$-prefixed subformulae.
  The above notions are canonically extended to sets of formulae.
  
  The key notion in autoepistemic logic are stable sets of beliefs grounded on the given premises. 
  These sets, called \emph{stable expansions}, are defined as the fixed points of the consequences of knowledge and belief.

  \begin{definition}
    Let $\Sigma \subseteq \allAutoepistemicFormulae(B)$. A set $\Delta\subseteq \allAutoepistemicFormulae$ is a \emph{stable expansion} of $\Sigma$ if it satisfies the condition
    $
      \Delta = \theorems{\Sigma \cup L(\Delta) \cup \neg L(\bar{\Delta})},
    $
    where $L(\Delta) := \{L\varphi \mid \varphi \in \Delta\}$ and $\neg L(\bar{\Delta}) := \{\neg L\varphi \mid \varphi \not\in \Delta\}$.
  \end{definition}
  
  
  The three main decision problems in the context of autoepistemic logic are
  deciding whether a given set of premises has a stable expansion, and 
  deciding whether a given formula in contained in at least one (resp.\ all) stable expansion.
  As we are to study the complexity of these problems for finite restricted sets $B$ of Boolean functions, we formally define 
  the \emph{expansion existence problem} as 
  \problemdef{$\EXP(B)$}
    {A set $\Sigma \subseteq \allAutoepistemicFormulae(B)$}
    {Does $\Sigma$ have a stable expansion?}
  and the \emph{brave (resp.\ cautious) reasoning problems} as 
  \problemdef{$\MEMB(B)$ (resp. $\MEMC(B)$)}
    {A set $\Sigma \subseteq \allAutoepistemicFormulae(B)$, a formula $\varphi \in \allAutoepistemicFormulae(B)$}
    {Is $\varphi$ contained in some (resp.\ any) stable expansion of $\Sigma$?}
  
  A central tool for the study of the computational complexity of the above problems is the following finite characterization of stable expansions given by Niemel\"a~\cite{nie91}.
  
  \begin{definition} \label{def:full-set}
    For a set $\Sigma \subseteq \allAutoepistemicFormulae$, a set $ \Lambda \subseteq \SF^L(\Sigma) \cup \neg \SF^L(\Sigma)$ is \emph{$\Sigma$-full} if for each $L\varphi \in \SF^L(\Sigma)$,
    \begin{enumerate}
      \item $\Sigma \cup \Lambda \models \varphi$ iff $L \varphi \in \Lambda$,
      \item $\Sigma \cup \Lambda \not\models \varphi$ iff $\neg L \varphi \in \Lambda$.
    \end{enumerate}
  \end{definition}
  
  \begin{lemma}[\cite{nie91}]\label{lem:correspondence-exp-fullset}
    Let $\Sigma \subseteq \allAutoepistemicFormulae$.
    \begin{enumerate}
      \item Let $\Lambda$ be a $\Sigma$-full set, then for every $L\varphi \in \SF^L(\Sigma)$ either $L\varphi \in \Lambda$ or $\neg L\varphi \in\Lambda$.
      \item The stable expansions of $\Sigma$ and $\Sigma$-full sets are in one-to-one correspondence.
    \end{enumerate}
  \end{lemma}
  
  To be more precise, 
  say that a formula is \emph{quasi-atomic} if it is atomic or else begins with an $L$.
  Further denote by $\SF^q(\varphi)$ the set of all maximal quasi-atomic subformulae of $\varphi$ 
  (in the sense that a quasi-atomic subformula is maximal if it is not a subformula of another quasi-atomic subformula of $\varphi$).
  Write $\SE(\Lambda)$ for the stable expansion corresponding to $\Lambda$ and say that $\Lambda$ is its \emph{kernel}.
  
  \begin{definition}\label{def:modelsL}
    Let $\Sigma \subseteq \allAutoepistemicFormulae$ and let $\varphi \in \allAutoepistemicFormulae$.
    We define the consequence relation $\models_L$ recursively as 
    $
      \Sigma \models_L \varphi 
      \iff
      \Sigma \cup \mathrm{SB}(\varphi) \models \varphi,
    $
    where $\mathrm{SB}(\varphi):= \{ L\chi \in \SF^q(\varphi) \mid \Sigma \models_L \chi \} \cup \{\neg L \chi \mid L\chi \in \SF^q(\varphi), \Sigma \not\models_L \chi \}$.
  \end{definition}
  
  The point in defining the consequence relation $\models_L$ is that, once a $\Sigma$-full set has been determined, it describes membership in the stable expansion corresponding to $\Lambda$.
 
  \begin{lemma}[\cite{nie91}]\label{lem:eqv_membership_SE}
    Let $\Sigma \subseteq \allAutoepistemicFormulae$, let $\Lambda$ be a $\Sigma$-full set and $\varphi \in \allAutoepistemicFormulae$.
    It holds that $\Sigma \cup \Lambda \models_L \varphi$ iff $\varphi \in \SE(\Lambda)$.
  \end{lemma}

\section{Complexity Results} \label{sect:complexity_results}

  The complexity of the before defined decision problems for autoepistemic logic has already been investigated by Niemel\"{a}~\cite{nie91} and Gottlob~\cite{gottlob:92}. Niemel\"{a}~\cite{nie91} proved that in order to show that a set $\Sigma$ has a stable expansion, we may guess a subset $ \Lambda \subseteq \SF^L(\Sigma) \cup \neg \SF^L(\Sigma)$ nondeterministically and check that it is  \emph{$\Sigma$-full} (see Definition \ref{def:full-set} and Lemma \ref{lem:correspondence-exp-fullset}). Thus the problem of deciding whether $\Sigma$ has a stable expansion is nondeterministically Turing-reducible to the propositional implication problem. Thus he proved that the extension existence problem is in $\SigmaPtwo$. According to Definition \ref{def:modelsL} and Lemma \ref{lem:eqv_membership_SE} the problem of  deciding whether there exists a stable expansion $\Sigma$ containing a given formula $\phi$ can be solved with a polynomial number of calls to an $\NP$-oracle by a nondeterministic Turing reduction as follows. Guess a subset $\Lambda$; check that $\Lambda$ is $\Sigma$-full and check that $\phi\in \SE(\Lambda)$. Therefore, the brave reasoning problem is in $\SigmaPtwo$, whereas the cautious reasoning problem is in $\PiPtwo$.
   Corresponding hardness results were obtained by  Gottlob~\cite{gottlob:92}. More precisely he obtained completeness  results for the special case   $B=\{\land,\lor,\neg\}$. 
   
  We investigate here the complexity of these problems for every $B$. Observe that the upper bounds, \emph{i.e.}, membership in $\SigmaPtwo$ (resp.\ $\PiPtwo$) still hold for any $B$. In order to classify the complexity for the infinitely many cases of $B$ we will make use of Post's lattice as follows: Suppose that $B\subseteq [B']$ for some finite sets $B,B'$ of Boolean functions. Then every function in $B$ can be expressed as a composition of functions from $B'$; in other words: for every $f\in B$ there is a propositional formula $\phi_{f}$ over connectives from $B'$ defining $f$, and every $\allAutoepistemicFormulae(B)$-formula can be transformed into an equivalent $\allAutoepistemicFormulae(B')$-formula. If moreover in the formulae $\phi_{f}$ every free variable appears only once (in this case we say that $\phi_{f}$ is a \emph{small} formula for $f$; and in the proofs below we will see that we can always construct such small formulae), then the transformation of a $\allAutoepistemicFormulae(B)$-formula $\psi$ into an equivalent $\allAutoepistemicFormulae(B')$-formula $\psi'$ is efficient in the sense that the length of $\psi'$ can be bounded by a polynomial  in the length of  $\psi$. Thus,  the  upper bound for the complexity of $\EXP(B')$ yields an upper bound for the complexity of $\EXP(B)$, and a lower bound for the complexity of $\EXP(B)$ yields a lower bound for the complexity of $\EXP(B')$. If $[B]=[B']$ then $\EXP(B)$ and $\EXP(B')$ are of the same complexity (w.r.t.~logspace reductions). Thus, the complexity of $\EXP(B)$ is determined by the clone $[B]$. This already brings some structure into the infinitely many problems $\EXP(B')$.

  We next note that we may w.l.o.g.\ assume the availability of the Boolean constants.%
  \setlength{\columnsep}{20pt}
  \begin{wrapfigure}[13]{r}{0.5\textwidth}
    \centering
    \includegraphics[width=\linewidth]{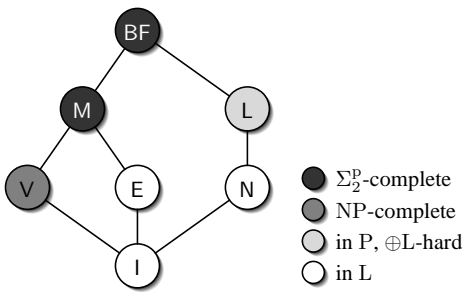}
    \caption{%
      \label{fig:post's-lattice}%
      Relevant clones and their inclusion structure;
      the shading indicates the complexity of $\EXP(B)$.
    }
  \end{wrapfigure}
  
  \begin{lemma} \label{lem:always-constants}
    Let $\mathcal{P}$ be any of the problems $\EXP$, $\MEMC$, or $\MEMB$.
    Then $\mathcal{P}(B) \equivlogm \mathcal{P}(B\cup \{\false,\true\})$ for all finite sets $B$ of Boolean functions.
  \end{lemma}
  
  \begin{proof}
    For the nontrivial direction,
    let $\Sigma \in \allAutoepistemicFormulae$.
    We map $\Sigma$ to $\Sigma':=\Sigma[\true/t,\false/Lf] \cup \{t\}$, 
    where $t$ and $f$ are fresh propositions.
    Then the stable expansions of $\Sigma'$ and $\Sigma$ are in one-to-one correspondence, 
    as any expansion of $\Sigma'$ includes $t$ and $\neg Lf$.
  \end{proof}

As a consequence of Lemma \ref{lem:always-constants}, it suffices to consider the clones of the form $[B\cup\{0,1\}]$ (as can be seen immediately from the list of clones given in \cite{bcrv03}). These are the seven clones $\CloneI$, $\CloneN$, $\CloneV$,  $\CloneE$,  $\CloneL$,  $\CloneM$, and $\CloneBF$ (see Fig.~\ref{fig:post's-lattice}).
All other cases will have the same complexity of these, by the explanations above. 
Before we start proving our classification, we note one further observation:

  \begin{lemma} \label{lem:inconsistent-expansion}
    For every set $\Sigma \subseteq \allAutoepistemicFormulae$, 
    $\allAutoepistemicFormulae$ is a stable expansion of $\Sigma$ iff $\Sigma \cup \SF^L(\Sigma)$ is inconsistent.
  \end{lemma}

  \begin{proof}
    Suppose that $\allAutoepistemicFormulae$ is a stable expansion of $\Sigma$ and let $\Lambda$ denote its kernel.
    Then $\Sigma \cup \Lambda \models_L \false$ by virtue of Lemma~\ref{lem:eqv_membership_SE}.
    As $\Sigma \cup \Lambda \models_L \false$ iff $\Sigma \cup \Lambda \models \false$, we obtain  $\Lambda=\SF^L(\Sigma)$ 
    (notice that $\{L\chi \mid L\chi \in \SF^q(\false)\}=\emptyset$, cf.\ Definition~\ref{def:modelsL}).
    In conclusion, $\Sigma \cup \SF^L(\Sigma)$ must be inconsistent.
    Conversely suppose that $\Sigma \cup \SF^L(\Sigma)$ is inconsistent. 
    Then so is $\theorems{\Sigma \cup L(\allAutoepistemicFormulae)}$.
    Consequently, any stable expansion must contain all autoepistemic formulae.
  \end{proof}

\subsection{Expansion Existence}

  \begin{theorem} \label{thm:ext}
    Let $B$ be a finite set of Boolean functions.
    \begin{itemize}
    \item  If $[B\cup\{0,1\}]$ is $\CloneBF$ or $\CloneM$ then $\EXP(B)$ is $\SigmaPtwo$-complete.
    \item If $[B\cup\{0,1\}]$ is $\CloneV$ then $\EXP(B)$ is $\NP$-complete.
    \item If $[B\cup\{0,1\}]$ is $\CloneL$ then $\EXP(B)$ is $\parityL$-hard and contained in $\P$.
    \item If $[B\cup\{0,1\}]$ is $\CloneE$ or $\CloneN$ or $\CloneI$ then $\EXP(B)$ is in $\L$ (solvable in logspace).
    \end{itemize}
  \end{theorem}
 
The proof of this theorem requires several propositions.

  \begin{proposition} \label{prop:exp-M}
    Let $B$ be a finite set of Boolean functions such that $\CloneM\subseteq [B]$.
    Then $\EXP(B)$ is $\SigmaPtwo$-complete.
  \end{proposition}
  
  \begin{proof}
    Let $B$ be a finite set of Boolean functions as required.
    We have to prove $\SigmaPtwo$-hardness. 
    Let $\varphi := \exists x_1 \cdots \exists x_n \forall y_1 \cdots \forall y_m \psi$ be a quantified Boolean formula in disjunctive normal form. In \cite{gottlob:92}, Gottlob shows that $\varphi$ is satisfied iff the set $\Sigma := \{L\psi, Lx_1 \leftrightarrow x_1,\ldots,Lx_n \leftrightarrow x_n \}$ has a stable expansion. The idea of our proof is to modify the given reduction to only use monotone connectives, thus showing that $\EXP(B)$ is $\SigmaPtwo$-hard for every finite set $B \subseteq \CloneM$.
    More precisely, we define 
    \[
      \psi':=\psi[\neg x_1/x_1',\ldots,\neg x_n/x_n',\neg y_1/y_1',\ldots,\neg y_m/y_m']        
    \]
    and let
    \[
      \Sigma' := \{L\psi'\} \cup \{ x_i \lor Lx_i', Lx_i \lor x_i' \mid 1 \leq i \leq n\} \cup \{ y_j \lor y_j' \mid 1 \leq j \leq m\}.
    \]
    Clearly, $\Sigma' \subseteq \allAutoepistemicFormulae(\{\land,\lor\})$. 
    Moreover, for every $1 \leq i \leq n$, we have that any stable expansion of $\Sigma$ contains either $Lx_i$ or $Lx_i'$ but not both:
    assume that $\Lambda$ is a $\Sigma'$-full set such that $Lx_i \in \Lambda$ and $Lx_i' \in \Lambda$. 
    Then, by definition of $\Sigma'$, $\Sigma' \cup \Lambda \not\models x_i, x_i'$, although $Lx_i, Lx_i' \in \Lambda$;
    a contradiction to  $\Lambda$ being $\Sigma'$-full.
    Otherwise, if $\Lambda$ were a $\Sigma'$-full set such that $\neg Lx_i \in \Lambda$ and $\neg Lx_i' \in \Lambda$, 
    then $\Sigma' \cup \Lambda \models x_i, x_i'$, a contradiction to $\Lambda$ being $\Sigma'$-full, because $Lx_i, Lx_i' \notin \Lambda$. 
    In conclusion, any $\Sigma'$-full set and equivalently any stable expansion contains either $Lx_i$ or $Lx_i'$ but not both.

    We show that $\Sigma'$ has a stable expansion if and only if $\varphi$ is valid.
    First suppose that $\Sigma'$ has a stable expansion $\Delta$. 
    Let $\Lambda$ denote its kernel. 
    As $\Sigma' \cup \SF^L{\Sigma'}$ is consistent, 
    we obtain that $\Delta \neq \allAutoepistemicFormulae$ from Lemma~\ref{lem:inconsistent-expansion}.
    By the argument above, either $Lx_i \in \Delta$ or $Lx_i' \in \Delta$, but not both. 
    Moreover, $L\psi' \in \Delta$, whence $\psi'$ must be derivable from $\Sigma' \cup \Lambda$ by Definition~\ref{def:full-set}. 
    Note that this implies that $\psi'$ is satisfied by all assignments setting either $y_i$ or $y_i'$ to $\true$; 
    in particular, by all assignments that assign a complementary value to $y_i$ and $y'_i$ for every $i$.
    Define a truth assignment $\sigma \colon \{x_i \mid 1 \leq i \leq n\} \to \{\false,\true\}$ from $\Lambda$
    such that $\sigma(x_i) := \true$ if $Lx_i \in \Lambda$, and $\sigma(x_i):= \false$ otherwise. 
    It follows that $\sigma \models \forall y_1 \cdots \forall y_m \psi$, thus $\varphi$ is valid.
    
    Now suppose that $\varphi$ is valid. 
    Then there exists an assignment $\sigma\colon \{x_i \mid 1 \leq i \leq n\} \to \{\false,\true\}$ 
    such that any extension of $\sigma$ to $y_1,\ldots,y_m$ satisfies $\psi$. 
    Let $\Lambda := \{ Lx_i,\neg Lx_i' \mid \sigma(x_i)=\true\} \cup \{\neg Lx_i, Lx_i' \mid \sigma(x_i)=\false\} \cup \{L\psi'\}$. 
    We claim that $\Lambda$ is  $\Sigma'$-full.
    If $Lx_i\in\Lambda$, then $\neg Lx_i' \in\Lambda$; hence $\{Lx'_i \lor x_i, \neg Lx_i'\}$ implies $x_i$. 
    Conversely, if $\Sigma'\cup\Lambda\models x_i$ then $\neg Lx_i'$ has to be in $\Lambda$, 
    because $x_i$ occurs in $L\psi'$ and the clause $Lx'_i \lor x_i$ only.
    From this, we obtain $Lx_i \in \Lambda$.
    Therefore, $\Sigma'\cup\Lambda\models x_i$ if and only if  $Lx_i\in\Lambda$. 
    From the definition of $\Lambda$ now follows that $\Sigma'\cup\Lambda\not\models x_i$ if and only if $\neg Lx_i\in\Lambda$. 
    The same holds for $x_i'$ for each $i$.
    Due to the construction of $\Lambda$, 
    the fact that the clause $y_i\lor y_i'$ enforces $y_i'$ to be assigned a value equal to or bigger than the one assigned to $\neg y_i$, 
    the definition of $\psi'$ and its monotonicity, 
    we also have $\Sigma' \cup \Lambda \models \psi'$. 
    Hence, following Definition~\ref{def:full-set}, $\Lambda$ is a $\Sigma'$-full set and, 
    by Lemma~\ref{lem:correspondence-exp-fullset}, $\Sigma'$ has a stable expansion.

    Finally, note that in any finite set of Boolean functions $B$ such that $\CloneM\subseteq [B]$, 
    conjunction and disjunction can be defined by small formulae, \emph{i.e.}, 
    there exist formulae $\phi_{\land}\equiv x\land y$ and $\phi_{\lor} \equiv x\lor y$ such that 
    $x$ and $y$ occur exactly once these formulae, see \cite{sch05}.
  \end{proof}

  We cannot transfer the above result to $\EXP(B)$ for $[B]=\CloneV$, 
  because we may not assume $\psi$ to be in conjunctive normal form. 
  But, using a similar idea, we can show that the problem is $\NP$-complete.
  
  \begin{proposition} \label{prop:exp-V}
    Let $B$ be a finite set of Boolean functions such that $[B\cup\{0,1\}]=\CloneV$. Then
    $\EXP(B)$ is $\NP$-complete.
  \end{proposition}
  
  \begin{proof}
    We first show that $\EXP(B)$ is efficiently verifiable, thus proving membership in $\NP$.
    Given a set $\Sigma \subseteq \allAutoepistemicFormulae$ and a candidate $\Lambda$ for a $\Sigma$-full set,
    substitute $L\varphi$ by the Boolean value assigned to by $\Lambda$ and call the resulting set $\Sigma'$. 
    Note that $\Sigma'$ is still equivalent to a set of disjunctions. 
    Therefore the conditions $\Sigma' \models \varphi$ if $L\varphi \in \Lambda$ and 
    $\Sigma' \not\models \varphi$ if $\neg L\varphi \in \Lambda$ can be verified in polynomial time, 
    for $\IMP(B) \in \P$ \cite{bemethvo08imp}.

    To show $\NP$-hardness, we reduce $\ThreeSAT$ to $\EXP(B)$ as follows.
    Let $\varphi:=\bigwedge_{1 \leq i \leq n} c_i $ with clauses $c_i = \ell_{i1} \lor \ell_{i2} \lor \ell_{i3}$, $1 \leq i \leq n$ be given and 
    let $x_1,\ldots,x_m$ enumerate the propositions occurring in $\varphi$.
    From $\varphi$ we construct the set 
    \[
      \Sigma := \{L c'_i \mid 1 \leq i \leq n\} \cup \{ x_i \lor Lx_i', Lx_i \lor x_i' \mid 1 \leq i \leq m\},
    \]
    where $c'_i=c_i[\neg x_1/x_1',\ldots,\neg x_m/x_m']$ for $1 \leq i \leq n$.
    Analogously to Proposition~\ref{prop:exp-M}, we obtain that for any stable expansion $\Delta$ of $\Sigma$ either $x_i \in \Delta$ or $x_i' \in \Delta$, but not both.
First, suppose that $\Delta$ is a stable expansion of $\Sigma$. 
    It is easily observed that $\Sigma \cup \SF^L(\Sigma)$ is consistent, 
    therefore $\Delta \neq \allAutoepistemicFormulae$.
    Let $\Lambda$ be the kernel of $\Delta$. 
    As $\Delta \neq \allAutoepistemicFormulae$ and $Lc_i' \in \Sigma$ for all $1 \leq i \leq n$, 
    Definition~\ref{def:full-set} implies that $\Sigma \cup \Lambda \models_L c_i'$ and 
    hence $\Sigma \cup \Lambda \models c_i$ for all $1 \leq i \leq n$.
    From this it follows that $\varphi$ is satisfied by the assignment $\sigma$ setting $\sigma(x_i)=\true$ iff $Lx_i \in \Delta$.
    
    Conversely, suppose that $\varphi$ is satisfied by the assignment $\sigma$. 
    Define the set 
    $\Lambda := \{ Lx_i,\neg Lx_i' \mid \sigma(x_i)= \true\} \cup \{ Lx_i',\neg Lx_i \mid \sigma(x_i)= \false\} \cup \{Lc'_i \mid 1 \leq i \leq n\}$. 
    As $\sigma \models c_i$ for any $1 \leq i \leq n$, we obtain that $\Sigma \cup \Lambda \models c'_i$.
    Concluding, $\Lambda$ is a $\Sigma$-full set.
  \end{proof}
  
  Next, we turn to the case $[B\cup\{0,1\}] = \CloneL$. 
  Say that an $L$-prefixed formula is \emph{$L$-atomic} if it is of the form $L\varphi$ for some atomic formula $\varphi$.
  
  \begin{lemma} \label{lem:exp-L-atomic}
    Let $\Sigma \subseteq \allAutoepistemicFormulae(\{\xor,\true\})$. 
    If $\SF^L(\Sigma)$ contains only $L$-atomic formulae, 
    then one can decide in polynomial time whether $\Sigma$ has a stable expansion.
  \end{lemma}
  \begin{proof}
    The idea is to use Gaussian elimination twice.
    Let $\Sigma$ be as required and suppose that $\Sigma$ consists of $m$ formulas. 
    Then the set $\Sigma$ can be seen as a system of linear equations and 
    thus written as $A\mathbf{x}=B\mathbf{y}+C$, where $\mathbf{x}=(x_1, \ldots, x_n)^\mathrm{T}$, 
    $\mathbf{y}=(Lx_1, \ldots, Lx_n)^\mathrm{T}$,   $A$ and $B$ are Boolean matrices having $m$ rows, and $C$
    is a Boolean vector.

    By applying Gaussian elimination to $A$ we obtain an equivalent system 
    $A'\mathbf{x}=B'\mathbf{y}+C'$ with an upper triangular matrix $A'$.
    Let $r$ denote the number of free variables in $A'\mathbf{x}$ and suppose w.l.o.g.\ that these are $x_1, \ldots, x_r$.
    By subsequently eliminating the variables $x_{r+1},\ldots,x_n$, 
    we arrive at a system $T$ equivalent to $\Sigma$ of the form:
    \begin{align*}
      \displaystyle
      \{x_i = f_i(x_1,\ldots, x_r) + g_i(Lx_1, \ldots, Lx_n)+c_i &\mid  r < i \leq n\} \; \cup \\
      \{0=g_i(Lx_1, \ldots, Lx_n)+c_i &\mid  n < i \leq m+r\},
    \end{align*}
    where for each $i$ the functions $f_i$ and $g_i$ are linear,  and $c_i$ is the constant  $0$ or $1$.
    
    Observe that $\Sigma\cup \SF^L(\Sigma)$ is inconsistent iff $T[Lx_1/1,\ldots Lx_n/1]$ has no solution. 
    In this case $\Sigma$ has $\allAutoepistemicFormulae$ as a stable expansion. 
    Let us now show how to construct a $\Sigma$-full set $\Lambda$ such that $\SE(\Lambda)\ne\allAutoepistemicFormulae$.
    
    Since the variables $x_1, \ldots, x_r$ are free, 
    they cannot be derived from $\Sigma\cup\Lambda$ whatever $\Lambda$ is. 
    The same occurs for every $i\geq r+1$ such that $f_i(x_1,\ldots, x_r)$ is not a constant function. 
    Suppose this is the case for $r+1 \leq i \leq s$.
    Then any $\Sigma$-full set has to contain $\neg Lx_j$ for $1 \leq j \leq s$. 
    Let $T'$ be the system obtained by considering all remaining equations while replacing $Lx_i$ with $0$ for each $1 \leq i \leq s$.
    For each equation in $T'$, the function $f_i$ (if present) is a constant function $\varepsilon_i$.
    Therefore $T'$ consists of the following equations:
    \begin{align*}\ \ \ 
      \{x_i = \varepsilon_i + g_i'(Lx_{s+1}, \ldots, Lx_n)+c_i &\mid  s < i \leq n\} \; \cup \\
      \{0=g_i'(Lx_{s+1}, \ldots, Lx_n)+c_i &\mid  n < i \leq m+r\}
    \end{align*}
    with $g_i'(Lx_{s+1}, \ldots, Lx_n) := g_i(0,\ldots,0,Lx_{s+1}, \ldots,$ $Lx_n)$ for $s < i \leq m+r$.
Thus, for every $\Lambda\subseteq \SF^L(\Sigma)\cup \neg\SF^L(\Sigma)$ such that 
$\{\neg Lx_1,\ldots, \neg Lx_s\}\subseteq\Lambda$ and every $i$, $\Sigma\cup\Lambda\models x_i$ (resp., $\Sigma\cup\Lambda\not\models x_i$) if and only if  $T'\cup\Lambda\models x_i$ (resp., $T'\cup\Lambda\not\models x_i$).

    \begin{claim} 
      Let $I$ and $J$ form a partition of $\{s+1,\ldots, n\}$.
      Then $(Lx_{s+1},\ldots,Lx_n)$ with $Lx_i = 0$ if $i \in I$ and $Lx_j = 1$ if $j \in J$
      is a solution of the system $T'[x_{s+1}/Lx_{s+1}, \ldots , x_{n}/Lx_{n}]$
      if and only if  
      $\Lambda=\{\neg Lx_1,\ldots, \neg Lx_s\} \cup \{\neg Lx_i \mid i\in I\} \cup \{Lx_j \mid j\in J\}$ 
      is a $\Sigma$-full set.
    \end{claim}
    
    \begin{proof}[Claim]
\def\endproof{{%
  \unskip\nobreak\hfil\penalty50%
 \hskip1em\hbox{}\nobreak$\triangle$%
 \parfillskip=0pt\par\endtrivlist\addpenalty{-100}%
}} 
    To prove the claim,
    let $\Lambda = \{\neg Lx_1,\ldots, \neg Lx_s\} \cup \{\neg Lx_i\mid i\in I\} \cup \{Lx_j\mid j\in J\}$ be a $\Sigma$-full set.
    Observe that $\Sigma\cup \Lambda$ is consistent and that either  $T'\cup \Lambda \models x_i$ or $T'\cup \Lambda \models \neg x_i$.
    Denote by $\lambda$ the truth assignment induced by $\Lambda$ on $\SF^L(\Sigma)$.
    Then, for every $i > s$, 
      $Lx_i\in\Lambda$ iff
      $\lambda(Lx_i)=1$ iff
      $T'\cup \Lambda\models x_i$ iff 
      $\varepsilon_i+g_i'\big(\lambda(Lx_{s+1}), \ldots, \lambda(Lx_n)\big)+c_i=1$; 
    and 
      $\neg Lx_i\in\Lambda$ iff
      $\lambda(Lx_i)=0$ iff 
      $T'\cup \Lambda\models\neg  x_i$ iff 
      $\varepsilon_i+g_i'\big(\lambda(Lx_{s+1}), \ldots, \lambda(Lx_n)\big)+c_i=0$.
    This means that for every $i$, we have $\varepsilon_i+g_i'\big(\lambda(Lx_{s+1}), \ldots, \lambda(Lx_n)\big)+c_i=\lambda(Lx_i)$.
    Therefore $\lambda$ is a solution of the system $\{Lx_i=\varepsilon_i+g_i(0, \ldots, 0,$ $Lx_{s+1}, \ldots, Lx_n) +c_i \mid s < i \leq n \}$, and hence of the system $T'[x_{s+1}/Lx_{s+1}, \ldots , x_{n}/Lx_{n}]$.
    
    Conversely, suppose that $\lambda$ is a solution of $T'[x_{s+1}/Lx_{s+1}, \ldots , x_{n}/Lx_{n}]$. 
    In particular, $\lambda$ satisfies $\lambda(Lx_i)=\varepsilon_i+g_i'\big(\lambda(Lx_{s+1}), \ldots, \lambda(Lx_n)\big) +c_i \mid s+1 \leq i \leq n\}$. 
    Set $\Lambda:=\{\neg Lx_1,\ldots, \neg Lx_s\} \cup \{\neg Lx_i \mid s+1\leq i\leq n, \lambda(x_i)=0 \} \cup \{ Lx_i\mid s+1\leq i\leq n, \lambda(x_i)=1 \}$.
    Then $T'\cup \Lambda$ is equivalent to $\{Lx_i=\varepsilon_i+g_i'\big(\lambda(Lx_{s+1}), \ldots, \lambda(Lx_n)\big) + c_i \mid s < i \leq n\} \cup \{0=g_i'\big(\lambda(Lx_{s+1}),  \ldots, \lambda(Lx_n)\big) + c_i\mid  n < i \leq m+r\} $. Therefore $T'\cup \Lambda\models x_i$ iff  $\lambda(Lx_i)=1$ and $T'\cup \Lambda\models  \neg x_i$ iff  $\lambda(Lx_i)=0$.
    Hence, $\Lambda$ is a $\Sigma$-full set. This proves the claim.
  \end{proof}
    
    We conclude that $\Sigma$ has a stable expansion iff $T'[x_{s+1}/Lx_{s+1}, \ldots , x_{n}/Lx_{n}]$ has a solution.
  \end{proof}
  
  Note that solving this last system by Gaussian elimination also gives the total number of possible $\Sigma$-full sets:
  the number of consistent stable expansions is equal to the number of solutions of $T'[x_{s+1}/Lx_{s+1}, \ldots , x_{n}/Lx_{n}]$; 
  (while testing for the inconsistent stable expansion can also be accomplished in polynomial-time as seen at the beginning of the proof).

  \begin{proposition} \label{prop:exp-L}
    Let $B$ be a finite set of Boolean functions such that $[B\cup\{0,1\}]=\CloneL$.
    Then $\EXP(B)$ is $\parityL$-hard and contained in $\P$.
  \end{proposition}

  \begin{proof}
    Let $B$ be as required and $\Sigma$ be a set of autoepistemic $B$-formulae.
    Then $\Sigma$ can be written in polynomial time as a set $\big\{c_k \xor \bigoplus_{i\in I_k} x_i \,\big|\, k \in \N, c_k \in \{\false,\true\} \big\}$ (see, \emph{e.g.}, \cite{bemethvo08imp}).
 
    We transform this set to $\Sigma'$ as follows:
    introduce a fresh variable $y_\phi$ for every non-atomic formula $\phi$ such that $L\phi\in\Sigma$;
    add the equations $y_\phi \equiv \phi$; and replace all occurrences of $L\phi$ by $Ly_\phi$. 
    We claim that the $\Sigma$-full sets and the $\Sigma'$-full sets are in one-to-one correspondence.
    This establishes the upper bound, because $\Sigma'$ satisfies the conditions of Lemma \ref{lem:exp-L-atomic}.

     To prove the claim, let $\Lambda\subseteq\SF^L(\Sigma)\cup \neg \SF^L(\Sigma)$. We give an inductive argument on the number of non-$L$-atomic formulae in $\Sigma$.
    To this end, choose an $L\varphi \in \SF^L(\Sigma)$ such that $\varphi$  does not contain $L$-prefixed subformulae.
    Define 
    \begin{align*}
      \Sigma_\varphi  &:= \Sigma[L\varphi/Ly_\varphi] \cup \{\varphi \equiv y_\varphi\}, \\
      \Lambda_\varphi &:= (\Lambda \setminus \{L\varphi, \neg L\varphi\}) \cup \{Ly_\varphi \mid L\varphi \in \Lambda\} \cup \{\neg Ly_\varphi \mid \neg L\varphi \in \Lambda\}.
    \end{align*}
    That is, $\Sigma_\varphi$ differs from $\Sigma$ in that we substituted one non-$L$-atomic subformula.

    Observe that $\Sigma \cup \Lambda \models \varphi $ if and only if $\Sigma_\varphi \cup \Lambda_\varphi \models y_\varphi$. 
    Therefore, since $L\varphi \in \Lambda$ if and only if $Ly_\varphi\in \Lambda_\varphi$, and 
    $\neg L\varphi \in \Lambda$ if and only if $\neg Ly_\varphi\in \Lambda_\varphi$, 
    it holds that $\Lambda$ is $\Sigma$-full if and only if $\Lambda_\varphi$ is $\Sigma_\varphi$-full.
    Repeating the above argument eventually yields $\Sigma'$, for which the existence of stable expansions 
    can be tested in polynomial time by Lemma~\ref{lem:exp-L-atomic}.
 
    It hence remains to establish $\parityL$-hardness.
    We give a reduction from $\IMP(B)$ for $[B\cup\{0,1\}]=\CloneL$,
    \emph{i.e.}, the problem to decide whether $\Gamma \models \psi$ for a given set $\Gamma$ of $B$-formulae and a given $B$-formula $\psi$.
    Since $\IMP(B)$ is $\parityL$-complete in this case, the proposition follows.
    For an instance $(\Gamma,\psi)$ of $\IMP(B)$, let $\Sigma:=\Gamma\cup \{L\psi\}$.
    Indeed, if $\Gamma \models \psi$, then $\Lambda:= \{L\psi\}$ is $\Sigma$-full; 
    and if $\Lambda:= \{L\psi\}$ is $\Sigma$-full, then $\Gamma \models \psi$.
    Thus, $\IMP(B) \leqcd \EXP(B)$ via the mapping $(\Gamma,\psi) \mapsto \Sigma$.
  \end{proof}
  
  \begin{proposition}\label{prop:exp-EN}
    Let $B$ be a finite set of Boolean functions such that $[B] \subseteq \CloneN$ or $[B] \subseteq \CloneE$.
    Then $\EXP(B)$ is solvable in $\L$. 
    It moreover holds that, for every set $\Sigma \subseteq \allAutoepistemicFormulae(B)$, there is at most one consistent stable expansion.
  \end{proposition}
  
  \begin{proof}
    Let $B$ be a finite set of Boolean functions such that $[B] \subseteq \CloneN$ and 
    let $\Sigma \subseteq \allAutoepistemicFormulae(B)$ be given.
    For $\Sigma$ to have a consistent stable expansion, 
    $\varphi$ has to be in $\Sigma$ for all $L\varphi \in \Sigma$, while
    $\varphi$ must not to be in $\Sigma$ for all $\neg L\varphi \in \Sigma$ or $L\neg\varphi \in \Sigma$.
    As $\Sigma \equiv \bigwedge \Sigma$, the result for $[B] \subseteq \CloneE$ follows from the above.
  \end{proof}

  
  
  The proof of Theorem \ref{thm:ext} now immediately follows from Propositions~\ref{prop:exp-M}--\ref{prop:exp-EN}. 
  Note that by Lemma~\ref{lem:always-constants} and the discussion following that lemma, this covers all cases and, hence, 
  Theorem \ref{thm:ext} gives a complete classification.
  
  From this theorem and its proof one can easily settle the complexity of the existence of a consistent stable expansion 
  as well as the complexity of the brave and cautious reasoning. 
  
  \begin{corollary}\label{cor:consistent-exp}
    For all finite sets $B$ of Boolean functions,
    the complexity of the problem to decide whether a set of 
    autoepistemic $B$-formulae has a consistent stable expansion is the same as 
    for the problem to decide the existence of a stable expansion.
  \end{corollary}
  \begin{proof}
    The corollary follows immediately from the proof of Theorem \ref{thm:ext}. Indeed, in each hardness proof (see Propositions \ref{prop:exp-M},  \ref{prop:exp-V} and \ref{prop:exp-L}) we have shown that the set of $B$-premises constructed in that proof, $\Sigma$ or $\Sigma'$, does not admit $ \allAutoepistemicFormulae$ as a stable expansion. Therefore, $\Sigma$ or $\Sigma'$ has a stable expansion iff it has a consistent stable expansion. This proves all the hardness results. 
    As for the upper bounds, Propositions~\ref{prop:exp-M} and \ref{prop:exp-V} are easily seen to extend to the existence of a consistent stable expansion.
    For the tractable cases $[B]\subseteq\CloneE$ and $[B]\subseteq\CloneN$, one can decide the existence of a consistent stable expansion in logarithmic space. This follows from the proof of Proposition \ref{prop:exp-EN}.
    Finally, for $[B]\subseteq \CloneL$, observe that the proof Proposition~\ref{prop:exp-L} actually allows to compute full sets
corresponding to consistent stable expansions in polynomial time.
  \end{proof}
  
\subsection{Brave and Cautious Reasoning}

  \begin{theorem}\label{thm:brave_cautious}
    Let $B$ be a finite set of Boolean functions.
    \begin{itemize}
      \item  If $[B\cup\{0,1\}]$ is $\CloneBF$ or $\CloneM$ then $\MEMB(B)$ is $\SigmaPtwo$-complete, whereas $\MEMC(B)$ is $\PiPtwo$-complete.
      \item If $[B\cup\{0,1\}]$ is $\CloneV$ then $\MEMB(B)$ is  $\NP$-complete, whereas $\MEMC(B)$ is $\co\NP$-complete.
      \item If $[B\cup\{0,1\}]$ is $\CloneL$ then $\MEMB(B)$ and $\MEMC(B)$ are $\parityL$-hard and in $\P$.
      \item If $[B\cup\{0,1\}]$ is $\CloneE$ or $\CloneN$ or $\CloneI$ then $\MEMB(B)$ and $\MEMC(B)$ are in $\L$.
    \end{itemize}
  \end{theorem}
  
  To prove Theorem~\ref{thm:brave_cautious}, we require two lemmas that provide upper bounds on the complexity of $\MEMB(B)$ and $\MEMC(B)$ via reduction to the expansion existence problem.
  
  \begin{lemma} \label{lem:memc-to-exp}
    Let $B$ be a finite set of Boolean functions such that $[B]=\CloneL$. 
    Then $\MEMB(B) \leqlogm \EXP(B)$.
  \end{lemma}
  
  \begin{proof}
    Let $B$ be a finite set of Boolean functions such that $[B]=\CloneL$.
    Given $\Sigma \subseteq \allAutoepistemicFormulae(B)$ and $\varphi \in \allAutoepistemicFormulae(B)$, map 
    the pair $(\Sigma,\varphi)$ to $\Sigma' := \Sigma \cup \{L\varphi \xor p \xor \true, Lp\}$, where $p$ is a fresh proposition.
    We claim that $\varphi$ is contained in a stable expansion of $\Sigma$ iff $\Sigma' \in \EXP(B)$.
    
    First suppose that $\varphi$ is contained in a stable expansion $\Delta$ of $\Sigma$ and let $\Lambda$ denote its kernel.
    We claim that $\Lambda':=\Lambda \cup \{ L\varphi, Lp\}$ is $\Sigma'$-full: 
    \begin{itemize}
      \item $\Sigma' \cup \Lambda' \models \varphi$, because $\Sigma \cup \Lambda \models_L \varphi$;
      \item $\Sigma' \cup \Lambda' \models p$, because $\Sigma \cup \{L\varphi, L\varphi \xor p \xor \true\} \models p$;
      \item for all $L\psi \in \Lambda$, we have 
      $\Sigma' \cup \Lambda' \equiv \Sigma \cup \Lambda \cup \{ L\varphi, L\varphi \xor p \xor \true, Lp \}\models_L \psi$; 
      whereas for all $\neg L\psi \in \Lambda$, we still have 
      $\Sigma' \cup \Lambda' \equiv \Sigma \cup \Lambda \cup \{ L\varphi, L\varphi \xor p \xor \true, Lp \} \not\models_L \psi$.
    \end{itemize}
    Hence, $\Sigma'$ has a stable expansion.
    
    Conversely, suppose that $\varphi$ is not bravely entailed. Hence $\Sigma$ does not have $\allAutoepistemicFormulae$ 
    as a stable expansion and $\neg L\varphi \in \Delta$ for all stable expansions $\Delta$ of $\Sigma$. Observe that
    $\Sigma'\cup \SF^L(\Sigma')= \Sigma\cup \SF^L(\Sigma)\cup \{L\varphi \xor p \xor \true, Lp\}\cup \{ L\varphi, Lp\}$ is consistent,
    therefore $\allAutoepistemicFormulae$ is not a stable expansion of $\Sigma'$.
    Hence, assume that $\Delta'$ is a consistent stable expansion of $\Sigma'$.
    Then either $Lp \in \Delta'$ or $\neg Lp \in \Delta'$. 
    In the former case, $\Delta'$ would also have to contain $L\varphi$,
    while $\varphi$ can not be derived. A contradiction to $\Delta'$ being a stable expansion of $\Sigma'$.
    In the latter case, we have that $\theorems{\Sigma' \cup L(\Delta') \cup \neg L(\bar{\Delta'})} 
    \supseteq \{\neg Lp, Lp\}$. Thus $\theorems{\Sigma' \cup L(\Delta') \cup \neg L(\bar{\Delta'})} = 
    \allAutoepistemicFormulae \supsetneq \Delta'$;a contradiction to $\Delta'$ being a stable expansion.
    We conclude that $\Sigma'$ does not posses any stable expansions.
  \end{proof}

  \begin{lemma} \label{lem:memb-to-exp}
    Let $B$ be a finite set of Boolean functions such that $[B]=\CloneL$.
    Then $\overline{\MEMC}(B) \leqlogm \EXP^*(B)$, where $\EXP^*(B)$ denotes the problem of deciding the existence of a consistent stable expansion.
  \end{lemma}
  \begin{proof}
    The proof is similar to the proof of Lemma~\ref{lem:memc-to-exp}.
    Let $B$ be a finite set of Boolean functions such that $[B]=\CloneL$.
    Given $\Sigma \subseteq \allAutoepistemicFormulae(B)$ and $\varphi \in \allAutoepistemicFormulae(B)$, map 
    the pair $(\Sigma,\varphi)$ to $\Sigma' := \Sigma \cup \{L\varphi \xor p, Lp\}$, where $p$ is a fresh proposition.
    We claim that $\varphi$ is contained in any stable expansion of $\Sigma$ iff $\Sigma' \not\in \EXP^*(B)$.
    
    First suppose that there exists a stable expansion $\Delta$ of $\Sigma$ that does not contain $\varphi$.
    Let $\Lambda$ denote its kernel. Then, for the same arguments as above, 
    $\Lambda':=\Lambda \cup \{ \neg L\varphi, Lp\}$ is a $\Sigma'$-full set.
    Conversely, suppose that $\varphi$ is contained in all stable expansions $\Delta$ of $\Sigma$.
    Let $\Delta'$ denote a consistent stable  expansion  of $\Sigma'$.
    If $Lp \in \Delta'$, then $\Delta'$ would also have to contain $\neg L\varphi$,
    while $\varphi$ can be derived. A contradiction to $\Delta'$ being a stable expansion of $\Sigma'$.
    Otherwise, if $\neg Lp \in \Delta'$, then $\Sigma' \cup L(\Delta') \cup \neg L(\bar{\Delta'})$ is 
    inconsistent---contradictory to $\Delta'$ being a consistent stable expansion.
    We conclude that $\Sigma'$ does not posses any consistent stable expansion.
  \end{proof}  
  
  \begin{proof}[{Proof of Theorem~\ref{thm:brave_cautious}.}]
    According to Lemma \ref{lem:always-constants} one can suppose w.l.o.g. that $B$ contains the two constants $0$ and $1$. Since $1$ belongs to all stable expansion, a set $\Sigma$ of $B$-premises has a stable expansion iff $1$ belongs to some  stable expansion  of $\Sigma$. Since $0$ does not belong to any consistent stable expansion,
    a set $\Sigma$ of $B$-premises has no consistent stable expansion iff $0$ belongs to any stable expansion  of $\Sigma$. 
    Therefore, the lower bounds follow from Theorem \ref{thm:ext} and Corollary~\ref{cor:consistent-exp}. 
    
    
    As for the upper bounds,
    membership in $\SigmaPtwo$ and $\PiPtwo$ in the general case follows from the discussion preceding Theorem~\ref{thm:ext}.
    
    For $[B]\subseteq \CloneV$, the proof of Proposition~\ref{prop:exp-V} shows 
    that, given $\Sigma \subseteq \allAutoepistemicFormulae(B)$, we can compute a $\Sigma$-full set $\Lambda$ in $\NP$ resp.\ $\co\NP$.
    By Lemma~\ref{lem:eqv_membership_SE}, it remains to check whether $\Sigma \cup \Lambda \models_L \varphi$.
    To this end, we nondeterministically guess a set $T \subseteq \SF^q(\varphi)$,
    verify that $\Sigma \cup \Lambda \cup \{L\chi \mid \chi \in T \} \cup \{\neg L \chi \mid \chi \in \SF^q(\varphi) \setminus T\} \models \varphi$, and recursively check that 
    \begin{itemize}
      \item $\Sigma \cup \Lambda \models_L \chi$ for all $\chi \in T$,
      \item $\Sigma \cup \Lambda \not\models_L \chi$ for all $\chi \in \SF^q(\varphi) \setminus T$.
    \end{itemize}
    This recursion terminates after at most $|\varphi|$ steps as $|\SF^q(\varphi)| \leq \SF(\varphi) \leq |\varphi|$ and
    $\Sigma \cup \Lambda \models_L \chi$ iff $\Sigma \cup \Lambda \models \chi$ for all for all propositional formulae $\chi$. 
    The above hence constitutes a polynomial-time Turing reduction to the implication problem for propositional $B$-formulae.
    As implication testing for $B$-formulae is in $\P$, we obtain that $\Sigma \cup \Lambda \models_L \varphi$ is polynomial-time decidable;
    thence $\MEMB(B) \in \NP$ and $\MEMC(B) \in \co\NP$.

    For $[B] \subseteq \CloneN$ and $[B] \subseteq \CloneE$,
    the proof of Proposition~\ref{prop:exp-EN} 
    shows that, given $\Sigma \subseteq \allAutoepistemicFormulae(B)$, 
    computation of a $\Sigma$-full set $\Lambda$ can be performed in $\L$, 
    while deciding $\Sigma \cup \Lambda \models_L \varphi$ reduces to testing whether 
    $\Sigma \cup \Lambda \models \psi$ for the (unique) atomic subformula $\psi \in \SF(\varphi)$.
    
    Finally, for $[B\cup\{0,1\}] = \CloneL$, the claim follows from Lemmas~\ref{lem:memc-to-exp} and \ref{lem:memb-to-exp}, 
Proposition \ref{prop:exp-L} and Corollary \ref{cor:consistent-exp}.
  \end{proof}

\section{Counting Complexity} \label{sect:counting}

  Besides deciding existence of stable expansions or entailment of formulae, 
  another natural question is concerned with the total number of stable expansions of a given autoepistemic theory. 
  We define the counting problem for stable extensions as
  \problemdef{$\#\EXP(B)$}
    {A set $\Sigma \subseteq \allAutoepistemicFormulae(B)$}
    {The number of stable expansions of $\Sigma$.}  
  The complexity of this problem is classified by the following theorem.
    
  \begin{theorem} \label{thm:count}
    Let $B$ be a finite set of Boolean functions.
    \begin{itemize}
      \item  If $[B\cup\{0,1\}]$ is $\CloneBF$ or $\CloneM$ then $\#\EXP(B)$ is $\SharpCoNP$-complete.
      \item If $[B\cup\{0,1\}]$ is $\CloneV$ then $\#\EXP(B)$ is  $\SharpP$-complete.
      \item If $[B\cup\{0,1\}] \subseteq \CloneL$ or $[B \cup \{0,1\}] \subseteq \CloneE$ then $\#\EXP(B)$ is in $\FP$.   
    \end{itemize}
  \end{theorem}
  
  \begin{proof}
    We first prove the lower bounds.
    It is easily observed that the reduction given in the proof of Lemma~\ref{lem:always-constants} is parsimonious. For the claimed lower bounds it hence suffices to prove the $\SharpCoNP$-hardness of $\#\EXP(B)$ for $[B]=\CloneM$ and 
    the $\SharpCoNP$-hardness of $\#\EXP(B)$ for $[B]=\CloneV$.
    For the former, notice that the reduction given in Proposition~\ref{prop:exp-V} is also a parsimonious reduction from $\#3\SAT$, which is  $\SharpP$-complete-complete via parsimonious reductions~\cite{val79}. For the latter, notice that the proof of Proposition~\ref{prop:exp-M} establishes a parsimonious reduction from the problem $\#\Pi_{1}\SAT$, which is $\SharpCoNP$-complete via parsimonious reductions \cite{dhk05}.

    We are thus left to prove the upper bounds.
    Let $B$ be a finite set of Boolean functions such that $[B]=\CloneBF$. 
    In the paragraph starting Section~\ref{sect:complexity_results}, it has been argued that the problem of deciding $\EXP(B)$ nondeterministically Turing-reduces to the propositional implication problem (see also~\cite{nie91}): given $\Sigma \subseteq \allAutoepistemicFormulae(B)$, guess a subset $\Lambda^+ \subseteq \SF^L(\Sigma)$ and verify that $\Lambda:=\Lambda^+ \cup \{\neg L\varphi \mid \varphi \in \SF^L(\Sigma), L\varphi \notin \Lambda^+\}$ is a $\Sigma$-full set using the conditions given in Definition~\ref{def:full-set}.
    It is thus clear that $\#\EXP$ is contained in $\#\mathord{\cdot}\P^{\NP}$, 
    as a Turing machine implementing the above algorithm can be build in a way 
    such that there is a bijection between its computation paths and the possible sets $\Lambda^+$.
    The first claim now follows from $\#\mathord{\cdot}\P^{\NP}=\SharpCoNP$~\cite{hevo95}.
    
    Next, let $B$ be such that $[B]=\CloneV$. 
    Then there exists a nondeterministic Turing machine $M$ such that the number of accepting path of $M$ on input $\Sigma \subseteq \allAutoepistemicFormulae(B)$ corresponds to the number of stable expansions of $\Sigma$ (cf.\ the proof of Proposition~\ref{prop:exp-V}).
    Hence, $\#\EXP(B) \in \SharpP$.
    
    Next, suppose that $[B] \subseteq \CloneL$ and let $\Sigma$ denote the given autoepistemic theory.
    Let $T'$ denote the system of linear equations obtained from $\Sigma$ 
    in the proofs of Lemma~\ref{lem:exp-L-atomic}.
    Then the number of consistent stable expansions of $\Sigma$ is equal to the number of solutions of the system 
    $T'[x_{s+1}/Lx_{s+1}, \ldots , x_{n}/Lx_{n}]$, which can be computed in polynomial time by Gaussian elimination.
    Moreover, $\allAutoepistemicFormulae$ is a stable expansion of $\Sigma$ iff $\Sigma \cup \SF^L(\Sigma)$ is inconsistent, which is 
    polynomial-time decidable. Hence, $\#\EXP(B) \in \FP$. 
    
    Finally, the case $[B] \subseteq \CloneE$ follows from the fact that for any $\Sigma \subseteq \allAutoepistemicFormulae(B)$ 
    an equivalent representation $\Sigma' \in \allAutoepistemicFormulae(\CloneI)$ can be computed efficiently.
  \end{proof}

\section{Conclusion} \label{sect:conclusion}
In this paper we followed the approach of Lewis to build formulae from a given finite set $B$ of allowed Boolean functions~\cite{lew79} and studied the complexity of the expansion existence, the brave (resp.\ cautious) reasoning problem, and the counting problem for stable expansions involving $B$-formulae.

We showed that for all sets of allowed Boolean functions, 
the computational complexity of the expansion existence and reasoning problems is divided into four presumably different levels 
(see Figure~\ref{fig:post's-lattice}): all three problems remain complete for classes of the second level of the polynomial 
hierarchy as long as the connectives $\land$ and $\lor$ can be expressed; if, otherwise, only disjunctions can be expressed 
the complexity drops to completeness for the first level of the polynomial hierarchy; in all remaining cases, the problems 
become tractable (either contained in $\L$ or contained in $\P$ and $\parityL$-hard).
We obtained a non-trivial polynomial-time upper bound for the case of not-unary affine functions. Note however that
the exact complexity of the problems in this case remains open.
This clone has also remained unclassified in a number of previous 
related works on different modal and non-monotonic logics \cite{bhss05b,thomas09}. 

As for the problem of counting the number of stable expansions, its computational complexity is trichotomic: 
$\SharpCoNP$-complete, $\SharpP$-complete, or contained in $\FP$. 
We think it is important to note that for our classification of counting problems the conceptually simple 
parsimonious reductions are sufficient, while for related classifications in the literature less restrictive 
(and more complicated) reductions such as subtractive or complementive reductions had to be used 
(see, \emph{e.g.}, \cite{dhk05,duhe08,bbcrsv09} and some of the results of \cite{hepi07}). Parsimonious reductions are not only the conceptually simplest ones
since they are direct analogues of the usual many-one reductions among languages.
They also form the strongest (strictest) type of reduction with a number of good properties, e.\,g.~all relevant counting classes are closed under parsimonious reductions (and not under the other mentioned types of reductions). 
Thus, one of the contributions of our paper is a natural counting problem complete in the class $\SharpCoNP$ 
under the simplest type of reductions. 

Future work, besides closing the gap for the clone $\CloneL$, should compare the classification obtained here to related classifications for  non-monotonic logics such as default logic \cite{bemethvo08} or circumscription \cite{thomas09}. It will be interesting to study if the embeddings between the three logics mentioned in the introduction obey the border between the clones.

\end{document}